\documentclass[10pt]{article}
\usepackage{amsmath}
\usepackage{amssymb}
\numberwithin{equation}{section}
\usepackage{amsthm}
\usepackage[colorlinks=true, pdfstartview=FitV, linkcolor=blue, citecolor=blue, urlcolor=blue]{hyperref}
\usepackage{caption}
\usepackage{subfigure}
\usepackage{graphicx}
\usepackage{mathrsfs}
\usepackage{url}
\theoremstyle{plain}
\newtheorem{proposition}{Proposition}[section]
\newtheorem{corollary}[proposition]{Corollary}
\newtheorem{lemma}[proposition]{Lemma}
\newtheorem{theorem}[proposition]{Theorem}
\theoremstyle{definition}
\newtheorem{definition}[proposition]{Definition}

\newtheorem{remark}[proposition]{Remark}

\DeclareMathOperator{\sgn}{sgn}

\newcommand{\R}{\mathbf{R}}
\newcommand{\C}{\mathbf{C}}

\newcommand{\abs}[1]{\left\lvert #1 \right\rvert}
\newcommand{\avg}[1]{\bigl\langle #1 \bigr\rangle}

\begin{document}
\title{Shock creation and Painlev\'{e} property of colliding peakons in the Degasperis-Procesi Equation}
\author{
  Jacek Szmigielski\thanks{Department of Mathematics and Statistics, University of Saskatchewan, 106 Wiggins Road, Saskatoon, Saskatchewan, S7N 5E6, Canada; szmigiel@math.usask.ca}
  \and  Lingjun Zhou\thanks{Department of Mathematics, Tongji University, Shanghai, P.R. China; zhoulj@tongji.edu.cn}}

\date{\today}
\maketitle
\begin{abstract}

The Degasperis-Procesi equation (DP) is one of several equations known to model important
nonlinear effects such as wave breaking and shock creation.  
It is, however, a special property of the DP equation that these two effects can be studied in an 
explicit way with the help of the multipeakon ansatz.  In essence this ansatz allows one to 
model wave breaking as a collision of hypothetical particles (peakons and antipeakons), 
called henceforth collectively multipeakons.  
It is shown that DP multipeakons have Painlev\'{e} property which implies a universal wave breaking behaviour, that multipeakons
can collide only in pairs, and that there are no multiple collisions other than, possibly simultaneous, collisions of peakon-antipeakon pairs at different locations.  Moreover, it is demonstrated that each peakon-antipeakon collision results in creation of a shock thus making possible a multi shock phenomenon.
\end{abstract}
\section{Introduction} \label{sec:intro}
The Degasperis-Procesi (DP) equation \cite{degasperis-procesi}
\begin{equation}\label{eq:DP}
u_t-u_{txx}+4uu_x=3u_x u_{xx}+uu_{xxx} 
\end{equation}
belongs to a class of one-dimensional wave equations which have attracted considerable 
attention over the last decade, following the most studied equation in this class, namely, the Camassa-Holm (CH) equation \cite{camassa-holm} 
\begin{equation*}
u_t-u_{txx}+3uu_x=2u_x u_{xx}+uu_{xxx}.  
\end{equation*}
Both these equations can be derived from the governing equations under the assumption of 
moderate amplitude \cite{Johnson, Constantin-Lannes}.  What makes them special is that, on one hand, both are Lax integrable, on the other, both 
exhibit wave breaking phenomenon not captured by linear theory or shallow water, small amplitude theory like the Korteweg-deVries equation.  The most relevant to this paper study of the breakdown of solutions for the CH equation was done by H.P. McKean \cite{McKean-breakdown, McKean-Asianbreakdown}.  In these works, it was argued that the breakdown of the CH waves is 
controlled by a kind of caricature of the higher dimensional vorticity, namely, $m=u-u_{xx}$ (see \cite{Majda-Bertozzi} ).   
In particular, it is the initial relative position of regions with positive $m$ versus negative $m$ that signals whether the breakdown will happen at some later, finite, time.  
One of the fascinating aspects of both the CH and DP equations is the existence of 
special solutions, peakons, which play the role of basic building blocks of the underlying full theory.  Peakons are 
a simple superposition of exponential terms 
\begin{equation*}
u(x,t)=\sum_{j=1}^n m_j(t)e^{-|x-x_j(t)|}
\end{equation*}
for which the function $m$ referred to earlier is $m=2\sum_{j=1}^n m_j \delta_{x_j}$.  
Were we to take the analogy with vorticity at its face value, $m$ for peakons 
could be viewed as a collection of point vortices, situated at $x_j(t)$s, of strength $m_j(t)$ each,  initially ordered in some fixed way, say, $x_1(0)<x_2(0)<\cdots< x_n(0)$.    
The case of CH peakons shows that if the strengths $\{m_k(0): 1\leq k\leq n\}$ are not of the same sign then {\sl collisions}  can occur, meaning that $x_j(t_c)=x_{j+1}(t_c)$ for some $j$  and some time 
$t_c$.  Each collision is accompanied by a blow-up of $m_j(t_c)$ and $m_{j+1}(t_c)$
resulting in the  derivative $u_x(t_c)$ becoming unbounded even though $u(t_c)$ remains bounded, in fact continuous.  Thus peakons can be used to test ideas about 
wave breaking, the advantage being that the peakon dynamics is described by a finite 
system of ODEs (see Section \ref{sec:basic} ).  
The analysis of the CH peakon collisions in this case was done in \cite{bss-moment} with the help of explicit formulas.  In short, the CH peakon problem can be solved by Stieltjes' 
formulas involving continued fractions \cite{bss-stieltjes}.  Moreover, the underlying 
boundary value problem is self-adjoint, in fact it is equivalent to an inhomogeneous string which remains isospectral under the CH flow.  

The case of the DP equation is superficially similar to the CH case.  However, 
deeper analysis shows a remarkable number of new features.  
For example, the associated spectral problem, termed a cubic string in 
\cite{ls-cubicstring}, is not self-adjoint and  
this has the immediate consequence that the inverse 
problem is by far more involved.  The peakon problem in the case of the positive 
measure, that is when all weights $m_j$s are positive, was solved explicitly in \cite{ls-cubicstring}.  
However, the generalization to the case of a signed measure $m$ is not straightforward 
since the spectral data  breaks up into several types depending on the degeneracy of the spectrum, 
as well as on certain coincidental phenomena of anti-resonances (eigenvalues $z_i, z_j$ 
pairing according to $z_i+z_j=0$).   
By contrast, the distinction between peakons for positive measure $m$ and peakons for the signed measure $m$ is less sharp for the CH case where the formulas for peakons can be analytically continued 
from former to latter.  This is not so for the DP case.  
This difficulty notwithstanding, in a way analogous to what happens in the CH case, the presence of  a collision signals an occurrence of wave breaking; in the DP
 context the connection between wave breaking and peakon collisions was studied earlier by H. Lundmark in \cite{lundmark-shockpeakons} for the case $n=2$ and further by the 
 present authors
 in \cite{sz1} for $n=3$.  Important questions of stability and general analytic results dealing 
 with DP peakons and the DP wave breaking have been addressed in  \cite{liu1, liu2, liu-zhaoyang, liu-escher}.  A considerable amount of work has been also done on 
 adapting numerical schemes to deal with the DP equation; we just mention a few: an operator splitting method of Feng and Liu \cite{liu-split}, or numerical schemes discussed by  Coclite, Karlsen and Risebro in \cite{coclite-karlsen-risebro}.

 The DP equation, in contrast to the CH equation, admits shock solutions (see \cite{coclite-karlsen-DPwellposedness,coclite-karlsen-DPuniqueness} for a general, very thorough, discussion). 
 It was H. Lundmark who introduced the concept of shockpeakons 
 \begin{equation*}
u(x,t)=\sum_{j=1}^n \{m_j(t)-s_j(t) \sgn (x-x_j(t))\}e^{-|x-x_j(t)|}, 
\end{equation*}
for which 
\begin{equation}\label{eq:shockm}
m=2\sum_{j=1}^n\{m_j \delta_{x_j}+s_j \delta_{x_j}'\}, 
\end{equation}
and showed that the solution describing a collision of two peakons 
has a unique entropy extension to shockpeakons.  He also 
hypothesized that this might be a general phenomenon valid also for $n>2$.  
We prove his conjecture.  More precisely we prove that 
the distributional limit of $m$ at the collision point  $t_c$ indeed produces shockpeakon data \eqref{eq:shockm} with positive shock strengths $s_j$ thus allowing 
a unique entropy weak extension (see Theorem \ref{thm:shocks} and Corollary \ref{cor:shocks}).  

Let us briefly describe our strategy.  
Instead of analyzing numerous spectral types we concentrate on analytic properties
of $x_j(t), m_j(t)$ as functions of $t$.  To this end we analyze the inverse spectral problem 
for the cubic string with the input data of a finite, signed measure.  We prove that each $x_j(t)$ must be a holomorphic function at $t_c$, while $m_j(t)$ is in general only meromorphic (Theorem \ref{thm:merocentral}).  Then we perform singularity analysis of the ODEs describing peakons \eqref{eq:b-ode-short} and prove their Painlev\'{e} property with the help of Theorems \ref{thm:merocentral} and \ref{thm:Pain} followed by a singularity analysis at the 
time of collisions described by Theorem \ref{thm:leadingcoef}.

The plan of the paper is as follows:
we review basic facts about the DP equation in Section \ref{sec:basic},
in Section \ref{sec:IP} we discuss the inverse problem for peakons of both signs
generalizing the uniqueness result known from the pure peakon case \cite{ls-cubicstring} and use this result to establish analytic properties of positions $x_j$s and masses $m_j$s.  
 In Section \ref{sec:blowup} we analyze the singular behaviour at the time of collisions and establish a universal singularity pattern according to which, in the leading term, only the time of the blowup depends on the initial conditions not the residue.  This fact is proven in 
Theorem \ref{thm:leadingcoef}.  We furthermore rule out triple collisions in Theorem \ref{thm:notriple}, and give an example of a simultaneous collision, in different 
positions, of two peakon-antipeakon pairs; finally in Section \ref{sec:shocks} we 
prove Theorem \ref{thm:shocks} stating that the distributional limit of colliding peakons is indeed a shockpeakon.

\section{Basic Facts about the DP equation}\label{sec:basic}
The nonlinear equation
\begin{equation}
  \label{eq:DP}
  u_t - u_{xxt} + 4u u_x = 3u_x u_{xx} + u u_{xxx},
\end{equation}
often written as
\begin{equation}
  \label{eq:DP-m}
  m_t + m_x u + 3m u_x = 0, \qquad m = u - u_{xx},
\end{equation}
was introduced by Degasperis and Procesi \cite{degasperis-procesi}
as an example of a nonlinear partial differential equation satisfying asymptotic integrability appearing in the family of
third order dispersive equations:
\begin{equation*}
u_t-\alpha^2 u_{xxt}+\gamma u_{xxx}+c_0 u_x=(c_1u^2+c_2u_x^2+c_3 uu_{xx})_x,
\end{equation*}
other examples of integrable equations in this family are
the Korteweg-deVries equation (KdV) and the Camassa-Holm (CH) equation.
Formal integrability for the DP equation was established by Degasperis, Holm and Hone \cite{degasperis-holm-hone} through the construction of
a Lax pair and a bi-Hamiltonian structure.
In particular, it was shown in \cite{degasperis-holm-hone} that the DP equation admits
the Lax pair:
\begin{equation}\label{eq:Lax}
(\partial_x-\partial_{xxx})\Psi=zm\Psi, \qquad
\Psi_t=[z^{-1}(1-\partial^2_x)+u_x-u\partial_x]\Psi.
\end{equation}
Moreover, one can impose additional boundary conditions provided they do not violate the compatibility of these equations.   One such a pair of boundary
conditions was introduced in \cite{ls-cubicstring}:
\begin{equation}\label{eq:asymptotics}
\Psi\rightarrow e^x,  \text{ as } x\rightarrow -\infty, \qquad  \Psi \text{ is bounded as   } x\rightarrow +\infty,
\end{equation}
where it was also shown that the spectrum of this boundary value problem
will remain time invariant (isospectral deformation).  It suffices for our purposes to
restrict our attention to the case in which $m$ is a finite discrete (signed) measure.  Thus for the remainder of the paper we will use the {\sl multipeakon} ansatz
\begin{equation}\label{eq:ansatz}
u(x,t)=\sum_{i=1}^nm_i(t)\, e^{-|x-x_i(t)|}
\end{equation}
where $x_1(0)<x_2(0)<\cdots<x_n(0)$ and $m_i(0)$ can have both positive and negative values.  This ansatz produces $m=2\sum_{i=1}^n m_i \delta_{x_i}$.
Moreover, with the proper interpretation of weak solutions to equation \eqref{eq:DP} we
can easily check that $u$ is a weak solution to \eqref{eq:DP} provided $x_i(t),m_i(t)$ satisfy the following ODEs
\begin{subequations}
  \label{eq:b-ode-short}
  \begin{align}
 & \dot x_k(t) = u(x_k)=\sum_{i=1}^n m_i(t)\, e^{-|x_k(t)-x_i(t)|}, \label{eq:b-ode-short-x}
  \qquad\\
  &\dot m_k(t) = 2 \, m_k(t) \, \avg{u_x(x_k)}=2\, m_k(t) \sum_{i=1}^n m_i(t) \sgn(x_k(t)-x_i(t)) \, e^{-|x_k(t)-x_i(t)|}, \label{eq:b-ode-short-m}
  \end{align}
\end{subequations}
where $\avg{f(x)}=\lim_{\epsilon \to 0^+}  \frac12 [f(x+\epsilon)+f(x-\epsilon)]$ is the average of $f$ at the point $x$.  We will refer to $m_j$s as {\sl masses} to emphasize
their role in the spectral problem. We also need a bit of terminology regarding
the phenomenon of breaking.  We will say that a {\sl collision } occurred
 at some time $t_c$  if  $x_i(t_c)=x_{i+1}(t_c)$ for some $i$. We can make this concept
 more geometric by introducing a configuration space in which to study peakon solutions,
 namely the sector $X=\{{\bf x} \in \R^n\ |x_1<x_2<\cdots<x_n\}$.  Then
 a collision corresponds to the solution $x_i$ hitting the boundary of $X$.

A very useful property of equations \eqref{eq:b-ode-short} is the
existence of $n$ constants of motion.  This follows readily from Theorem 2.10 in \cite{ls-cubicstring}.
\begin{lemma}\label{lem:constants}
$M_p\,(1\leq p\leq n)$, given by:
\[M_p=\sum_{I\in {[1,k]}\choose p}\left(\prod_{i\in I}m_i\right)\left(\prod_{j=1}^{p-1}(1-e^{x_{i_j}-x_{i_{j+1}}})^2\right)\]  are $n$ constants of motion of the system of equations \eqref{eq:b-ode-short}, where $\binom{[1,n]}{p}$ is the set of all $p$-element subsets
  $I=\left\{ i_1 < \dots < i_p \right\}$ of $\left\{ 1,\dots,n \right\}$.
\end{lemma}

\section{Inverse Problem for multipeakons}\label{sec:IP}
The boundary value value problem \eqref{eq:Lax} and \eqref{eq:asymptotics}
can be transformed to a finite interval boundary value problem, the cubic string
problem.  Indeed,
following \cite{ls-cubicstring}, the change of variables (Liouville transformation)
\begin{equation} \label{eq:Liouville}
y=\tanh \frac x2, \qquad \Psi(x)=\frac{2\phi(y)}{1-y^2}
\end{equation}
maps the DP spectral problem into the cubic string problem:
\begin{subequations}\label{eq:cubicstring}
\begin{align}
-\phi_{yyy}(y)&=zg(y)\phi(y), \qquad -1<y<1, \label{eq:acubicstring}\\
\phi(-1)&=\phi_y(-1)=\phi(1)=0,
\end{align}
\end{subequations}
where $g$ is the transformation of the measure $m$ induced by the Liouville transformation 
\eqref{eq:Liouville}.
Furthermore, as one can explicitly check, $g$ is also a finite signed measure and its support does not include the endpoints if the original measure $m$ is a finite signed measure.   More concretely, in this paper,
\begin{equation}\label{eq:g}
g=\sum_{i=1}^n g_i \delta_{y_i},    \quad -1<y_1<y_2<\dots<y_n<1,
\end{equation}
with weights $g_i\in \R$.
The inverse problem is studied with the help of two Weyl functions.  
\begin{definition} \label{def:WZ}
Let $\phi(y;z)$ denote the solution to the initial value problem \eqref{eq:acubicstring} with initial conditions $\phi(-1;z)=\phi_y(-1;z)=0, \, \phi_{yy}(-1;z)=1$.  The Weyl functions are ratios: 
\begin{equation*}
W(z)=\frac{\phi_y(1;z)}{\phi(1;z)}, \qquad Z(z)=\frac{\phi_{yy}(1;z)}{\phi(1;z)}.  
\end{equation*}
\end{definition}
These two functions encode spectral information needed to solve the inverse problem.
It is easy to verify that in the case of \eqref{eq:g}
both $W(z) $ and $Z(z)$ are rational functions which makes inversion algebraic.
However, in contrast to the pure peakon case $g_i>0$, the spectrum of the
boundary value problem \eqref{eq:acubicstring} is in general complex and
not necessarily simple.  This makes the inversion more challenging.
Regardless of the complexity of the spectrum though the Weyl functions undergo a simple evolution under the
DP flow.  Indeed,
using the second member of the Lax pair given by \eqref{eq:Lax} one can
find the time evolution of $W(z)$ and $Z(z)$.  To wit, using results
from Theorem 2.3 in \cite{sz1}
we obtain the following characterization of the time evolution of $W(z)$
and $Z(z)$.
\begin{theorem}
Let
\begin{equation*}
\frac{W(z)}{z}=\sum_j \sum_{k=1}^{d_j} \frac{b_j^{(k)}}{(z-\lambda_j)^k} +\frac1z,
\end{equation*}
be the partial fraction decomposition of $\frac{W(z)}{z}$, where $d_j$ denotes the algebraic degeneracy of the $j$-th eigenvalue.  Then the
DP time evolution implies:
\begin{enumerate}
\item[(1)]
\begin{equation*}
b^{(k)}_j=p^{(k)}_j(t)e^{\frac{t}{\lambda_j}},
\end{equation*}
where $p^{(k)}_j(t)$ is a polynomial in $t$ of degree $d_j-k$.
\item[(2)]$M_+\stackrel {\rm {def}} {=}\sum_{k=1}^n m_k e^{x_k}=
\sum_j \dot b^{(1)}_j$.
\item[(3)] \begin{equation*}
\dot W=\frac{W-1}{z} +M_+, \quad \dot Z=(W-1)M_++\dot W
\end{equation*}

\end{enumerate}
\end{theorem}
We immediately have:
\begin{corollary}\label{cor:WZflow} 
Under the DP flow the Weyl functions $W,Z$ are entire
functions of time.
\end{corollary}

The uniqueness result below plays a major role in the solution to the inverse problem.  
\begin{theorem} 
Suppose $\Phi: g\rightarrow \{W(z), Z(z)\}$ is the map that associates 
to the cubic string problem \eqref{eq:cubicstring} with a finite signed
measure $g$, the Weyl functions $W(z), Z(z)$.  
Then $\Phi$ is injective.  
\end{theorem} 
\begin{proof} 
The proof relies on remarks made in \cite{ls-invprob}.  We will construct 
a recursive scheme to solve the inverse spectral problem; given $W$ and $Z$ 
obtained from the map $\Phi$ we will reconstruct the finite, signed measure $g$ whose Weyl functions are $W$ and $Z$.  
More precisely, we will show that  the $y_j$'s and $g_j$'s in equation \eqref{eq:g} 
are uniquely determined from $W(z), Z(z)$.  First we recall that $W$ and $Z$ are 
constructed from solutions to the initial value problem (see Definition \ref{def:WZ}) 
\begin{subequations}\label{eq:stringIVP}
\begin{align}
-\phi_{yyy}(y)&=zg(y)\phi(y), \qquad -1<y<1, \\
\phi(-1)&=\phi_y(-1)=0, \quad \phi_{yy} (-1)=1
\end{align}
\end{subequations}
Masses $g_j$ are situated at $y_j, \, 1\leq j\leq n$ and for convenience let us set $y_0=0,\,  y_{n+1}=1$ and denote by $l_j=y_{j+1}-y_j$ the length of the interval $(y_j, y_{j+1})$.  
Then on each interval $(y_j, y_{j+1})$ 
the solution to (\ref{eq:stringIVP}) takes the form 
\begin{equation*}
  \phi(y) = \phi(y_{j+1}) + \phi_y(y_{j+1}) \, (y-y_{j+1})
            + \phi_{yy}(y_{j+1}-) \, (y-y_{j+1})^2/2,  \quad 0\leq j\leq n
\end{equation*}
and the condition of crossing $y_{j+1}$ is: continuity of $\phi$ and $\phi_y$ and 
the jump condition $\phi_{yy}(y_{j+1}+)-\phi_{yy}(y_{j+1}-)=-zg_{j+1} \phi(y_{j+1})$.  We establish, 
for example by an easy induction, 
\begin{equation} \label{eq:lphij} 
\left(\phi(y_{j+1}), \phi_y(y_{j+1}), \phi_{yy}(y_{j+1}-)\right)=(-z)^j \prod_{k=1}^j \frac{g_k l^2_{k-1}}{2}\,  \left(l^2_j/2, l_j, 1\right)+O(z^{j-1}), 
\end{equation} 
valid for $0\leq j\leq n$, with the convention that for $j=0$ the product equals $1$ and there is no remainder.  Likewise, 
\begin{equation}\label{eq:rphij} 
\left(\phi(y_{j}), \phi_y(y_{j}), \phi_{yy}(y_{j}+)\right)=(-z)^j \prod_{k=1}^j \frac{g_k l^2_{k-1}}{2}\,  \left(0, 0, 1\right)+O(z^{j-1}), 
\end{equation}
valid for $1\leq j\leq n$.

For $0\le j\le n$ we define 
$(w_{2j},z_{2j})= 
(\frac{\phi_y}{\phi} ,\frac{\phi_{yy}}{\phi})|_{y=y_{j+1}-}$
 and $(w_{2j-1},z_{2j-1})= 
(\frac{\phi_y}{\phi_{yy}},\frac{\phi}{\phi_{yy}})|_{y=y_j +}$.  
These quantities are essentially the left hand and the right hand analogs of 
Weyl functions introduced in Definition \ref{def:WZ} and correspond to shorter strings terminating at $y_{j+1}$ with no mass at the endpoint, or terminating at $y_j$ but with the mass $g_j$ at the end.  
Equation \eqref{eq:stringIVP} implies that the sequence $(w_{2j},z_{2j}, w_{2j-1},z_{2j-1})$ satisfies the recurrence relations 
  \begin{align} \label{eq:convergents}
    w_{2j-1}= \frac{w_{2j}}{z_{2j}}-l_j,
    \qquad 
    &
    z_{2j-1}=\frac{1}{z_{2j}}-l_j \frac{w_{2j}}{z_{2j}}+\frac{l_j^2}{2},
    \\
    w_{2j-2}= \frac{w_{2j-1}}{z_{2j-1}},
    \qquad 
    &
    z_{2j-2}=\frac{1}{z_{2j-1}}+zg_j; 
\end{align}
the iteration starts at $w_{2n}=W(z), z_{2n}=Z(z)$
and terminates at $w_{-1}, z_{-1}$. 
Moreover, based on equations \eqref{eq:lphij} and \eqref{eq:rphij}, we easily establish 

\begin{equation}\label{eq:seqasmpt} 
w_{2j}=\frac{2}{l_j} +O(\frac1z),\quad  z_{2j}=\frac{2}{l_j^2}+O(\frac1z), \quad w_{2j-1}=z_{2j-1}=O(\frac1z),  \quad \text{as } z\rightarrow \infty, 
\end{equation}
which implies that the quantities $\{l_j, g_j\}$ are 
determined in each step from the large $z$ asymptotics of terms known 
from the previous step.   Indeed, if we  
denote by $a^{(m)}$ the coefficient of $z^{-m}$
in the expansion of a 
holomorphic function $a(z)$ at $z=\infty$ we obtain the recovery formulas 
\begin{equation}
  \label{eq:recovery}
  l_j=\frac{2}{w_{2j}^{(0)}},
  \qquad
  g_j=-\frac{1}{z_{2j-1}^{(1)}}.  
\end{equation}
Thus we proved that given a pair of Weyl functions $W(z), Z(z)$ obtained from 
a cubic string problem \eqref{eq:stringIVP} with a finite, signed measure $g$,  
there exists a unique solution to the recurrence relations \eqref{eq:convergents} subject to \eqref{eq:seqasmpt} and thus a unique cubic string corresponding to $W(z), Z(z)$.  
\end{proof} 

We are now ready to state the central theorem of this section
\begin{theorem} \label{thm:merocentral} 
Let $\{x_j(t), m_j(t)\}, j=1, \dots, n$ be the positions and masses of the peakon ansatz \eqref{eq:ansatz} corresponding to an arbitrary signed measure $m=2 \sum_{j=1}^n m_j \delta _{x_j}$, satisfying peakon equations \eqref{eq:b-ode-short} on the time interval $(0,t_c)$ and suppose that a collision occurs 
at $t_c$.  Then the positions $x_1(t)\dots, x_n(t)$ are analytic functions at $t_c$, while the masses $m_1(t)\dots m_n(t)$ are given by meromorphic functions at $t_c$.  
\end{theorem} 
\begin{proof} Given the initial conditions $\{x_1(0)<x_2(0)<\dots< x_n(0)\}$ and \\
 $\{m_1(0), m_2(0), \dots, m_n(0)\}$ we set up the string problem \eqref{eq:acubicstring} after mapping $m(0)$ to $g(0)$.  This produces the Weyl functions $W(0), Z(0)$, which under the peakon flow evolve 
as entire functions of time in view of Corollary  \ref{cor:WZflow}.  We then set up the recursive scheme \eqref{eq:convergents} with $W(t), Z(t)$ as inputs.  At each stage of 
recursion only rational operations are involved and since the recursion is finite the formulas 
\eqref{eq:recovery} result in functions meromorphic in $t$.  Thus all $g_j, y_j$ are meromorphic in $t$.  For $t<t_c$ all distances $l_j>0$ and at $t_c$ some 
$l_i$ vanishes but all $l_j$ remain finite, because this is a finite string.  Hence $l_j(t)$ is regular at $t_c$ hence analytic there.  For a signed measure $g$ there are no bounds 
restrictions on individual $g_j$ so in general $g_j$ remains meromorphic at $t_c$.  
Mapping back to the real axis is afforded by $y=\tanh\frac x2 $; hence positions of 
individual masses are given by $x_j=\ln \frac{y_j+1}{y_j-1}$.  The only singular points of this map are for $y_j=\pm 1$ which means the end of the string or, after mapping the problem back to the 
real axis, $\pm \infty$.  However, based on results in \cite{sz1}, none of the masses can 
escape to $\pm \infty$ in finite time. So $ \frac{y_j+1}{y_j-1}$ is in the domain of analyticity of $\ln$ and hence the $x_j$s are analytic at $t_c$.  The relation between the measures $m$ and $g$ appearing in equations \eqref{eq:Lax} and \eqref{eq:acubicstring} is given 
by $m_j=\frac{(1-y_j^2)^2}{8} g_j$ which implies the claim since $g_j$ is meromorphic and 
$y_j$ analytic.  \end{proof} 

The above theorem establishes that the only singular points of solutions to the peakon 
ODE system \eqref{eq:b-ode-short-x} and \eqref{eq:b-ode-short-m} are poles.  
Since the inverse problem argument is valid for a fixed ordering $x_1<x_2<\cdots<x_n$ of 
masses,  the analytic continuation of masses and positions into the complex domain in $t$ will satisfy 
equations \eqref{eq:b-ode-short-x} and \eqref{eq:b-ode-short-m} in which $\sgn(x_k-x_i), e^{-\abs{(x_k-x_i)}}$ are 
replaced with $\sgn(k-i)$, $e^{-\sgn(k-i) (x_k-x_i)}$ respectively, to be consistent with the original ordering.  
It is for these equations that we note the {\sl absence of movable critical points} also known as {\sl Painlev\'{e} property} \cite{Ince}.  To facilitate the statement of the last theorem of this 
section we set $X_i=e^{x_i}, \, 1\leq i\leq n$ and rewrite the system \eqref{eq:b-ode-short-x} and \eqref{eq:b-ode-short-m} in new variables $\{m_i, X_i\}$.  
\begin{theorem} [Painlev\'{e} property] \label{thm:Pain} The system of 
differential equations 
  \begin{equation*}
 \dot X_k = X_k\sum_{i=1}^n m_i\, \left(\frac{X_i}{X_k}\right)^{\sgn(k-i)}, 
  \, \quad 
 \dot m_k= 2\, m_k \sum_{i=1}^n m_i \sgn(k-i) \, \left(\frac{X_i}{X_k}\right)^{\sgn(k-i)}
  \end{equation*}
  has the Painlev\'{e} property.  
\end{theorem}
\begin{proof} 
First we observe (using the variables of the proof of Theorem \ref{thm:merocentral}) 
that $X_i=\frac{y_i+1}{y_i-1}$, hence $X_i$s are meromorphic in $t$ because so are $y_i$s.  
The formulas for $X_i$s and $m_i$s obtained from the inverse problem are 
meromorphic in $t$ in the complex plane $\C$ and depend on $2n$ constants (spectral data 
consisting, in the generic case, of $n$ positions of poles  and $n$ residues of the Weyl function 
$W$),
which for the cubic string problem, in view of the ordering condition, are confined to an open set in $\C^{2n}$ by continuity of the inverse spectral map.  Relaxing that condition results in 
a solution depending on $2n$ arbitrary constants which comprises a general solution which is 
meromorphic in $t$ in the whole complex plane $\C$.  
\end{proof}  
In the remainder of the paper we will concentrate on the specific singularity structure 
at the time of collisions of peakons.  
\section{Blow-up behaviour}\label{sec:blowup}

We now proceed to establish several theorems on peakon collisions for DP equation. To begin with we recall the definition of a peakon collision briefly discussed in the introduction. We call $t_c$ the \textit{collision} time if there exists some $i$ such that
\begin{equation}\label{def:collision}
\lim_{t\to t_c^-}x_i(t)=\lim_{t\to t_c^-}x_{i+1}(t),
\end{equation} where $x_i(t)$s are the position functions in the ansatz \eqref{eq:ansatz}. Equivalently, we say that the $i$-th peakon collides with the $(i+1)$-th peakon at the time $t_c$. If there exist more than two position functions being identical at $t_c$ then we will say that a \textit{multiple collision} happens at $t_c$.

In this section, we describe the behaviour of the peakon dynamical system \eqref{eq:b-ode-short} in the neighbourhood of a collision time $t_c$.

To this end we need to study a special skew-symmetric $n\times n$ real matrix $A_n$ given by
\begin{equation}\label{eq:MatrixForm}
A_n=\left[\sgn(i-j)a_{ij}\right]
\end{equation}
whose entries satisfy $a_{ij}=a_{ji}$ and 
\begin{equation}\label{eq:MatrixCon}
a_{ij}=a_{il}a_{lj}\neq0,\quad\text{for all $0<i<l<j<n$.}
\end{equation} The following propositions hold for such a matrix.  
\begin{lemma}There exists a matrix $P$ with $\det P=1$ such that $P^TA_nP=B_n$, where
\[B_n=\left[\begin{array}{ccccccc}
0&-a_{12}\\
a_{12}&0\\
&&0&-a_{34}\\
&&a_{34}&0\\
&&&&\ddots\\
&&&&&0&-a_{n-1\,n}\\
&&&&&a_{n-1\,n}&0
\end{array}\right],\quad \text{if $n$ is even},\] or \[B_n=\left[\begin{array}{cccccccc}
0\\
&0&-a_{23}\\
&a_{23}&0\\
&&&0&-a_{45}\\
&&&a_{45}&0\\
&&&&&\ddots\\
&&&&&&0&-a_{n-1\,n}\\
&&&&&&a_{n-1\,n}&0
\end{array}\right],\quad \text{if $n$ is odd}.\]
\end{lemma}
\begin{proof}

  The conclusion is trivial for $n=1,2$. We assume the conclusion to hold for $n-2$; to show that it holds for $n$ we divide $A_{n}$ into four block submatrices 
  \begin{equation}\label{eq:An}
  A_{n}=\left[\begin{array}{cc}A_{n-2} & -B\\ B^T&C\end{array}\right],\end{equation}
   where $C=\left[\begin{array}{cc}0&
-a_{n-1\,n}\\a_{n-1\,n}&0\end{array}\right]$.   
Let us set $P_1=\left[\begin{array}{cc}I_{n-2}&0\\ -C^{-1}B^T&I_2\end{array}\right]$, 
then a direct computation shows that \[P_1^TA_{n} P_1=\left[\begin{array}{cc}A_{n-2}-BC^{-1}B^T& 0\\ 0&C\end{array}\right].\] In view of condition \eqref{eq:MatrixCon}, $B$ can be written as  $(\mathbf a,a_{n-1\,n}\mathbf a)$, where $\mathbf a =(a_{1\,n-1},a_{2\,n-1},\ldots,a_{n-2\,n-1})^T$.  It is now elementary to verify that $BC^{-1}B^T=0$. By the induction hypothesis there exists a matrix $P_2$ with $\det P_2=1$ such that $P_2^TA_{n-2} P_2=B_{n-2}, $ hence if we set
\[P=P_1\left[\begin{array}{cc}P_2& 0\\ 0&I_2\end{array}\right], \]  the conclusion follows.

\end{proof}

\begin{corollary}\label{lem:EvenDet}
If $n=2k$ then $\det A_{2k}=\displaystyle{\prod_{i=1}^ka_{2i-1\,2i}^2}>0$. If $n=2k+1$ then the rank of $A_{2k+1}$ is $2k$.
\end{corollary}

\begin{lemma}\label{lem:OddDet}
Let $E=(1,1,\ldots,1), \, n=2k+1$ and all entries satisfy $0<a_{ij}\leq 1$.  Then the rank of the matrix $\left[\begin{array}{c}E\\A_{2k+1}\end{array}\right]$ is $2k+1$.
\end{lemma}
\begin{proof}

Let $n$ be any odd number.  It suffices to show that the determinant of the matrix
\[\tilde A_{n}=\left[\begin{array}{ccccccc}
1&1&1&\cdots&\cdots&\cdots&1\\
a_{12}&0&-a_{23}&-a_{24}&\cdots&\cdots&-a_{2\,n}\\
a_{13}&a_{23}&0&-a_{34}&\cdots&\cdots&-a_{3\,n}\\
\vdots&&a_{34}&\ddots&\ddots&&\vdots\\
\vdots&&&\ddots&\ddots&-a_{n-2\,n-1}&-a_{n-2\,n}\\
\vdots&&&&a_{n-2\,n-1}&0&-a_{n-1\,n}\\
a_{1\,n}&a_{2\,n}&a_{3\,n}&\cdots&\cdots&a_{n-1\,n}&0
\end{array}\right] \] is positive.
For $n=3$ direct computation shows that \[\det\tilde A_3=a_{13}(1-a_{23})+a_{23}^2>0.\] We assume now that the conclusion holds for $n-2$.  We will show that it also holds for $n$. First, we divide $\tilde A_{n}$ into four submatrices by \[\tilde A_{n}=\left[\begin{array}{cc}\tilde A_{n-2}& -\hat B\\ B^T&C\end{array}\right],\] where \[C=\left[\begin{array}{cc}0&
-a_{n-1\,n}\\a_{n-1\,n}&0\end{array}\right],\, B=\left[\begin{array}{cc}a_{1\,n-1}&a_{1\,n}\\\mathbf{b} &a_{n-1\,n}\mathbf{b} \end{array}\right],\, \hat B=\left[\begin{array}{cc}-1&-1\\\mathbf{b} &a_{n-1\,n}\mathbf{b} \end{array}\right],\] and $\mathbf{b} =(a_{2\,n-1},a_{3\,n-1},\ldots,a_{n-2\,n-1})^T$. Since $C$ is invertible we can factor $\tilde A_n$ into the product of 
upper and lower block triangular matrices as follows: 
\begin{equation*}
\begin{bmatrix}\tilde A_{n-2}& -\hat B\\ B^T&C\end{bmatrix}=\begin{bmatrix}
\tilde A_{n-2}+\hat B C^{-1} B^T& -\hat BC^{-1} \\ 0&I_2 \end{bmatrix} \begin{bmatrix}I_{n-2}& 0\\ B^T&C \end{bmatrix}. 
\end{equation*}
Hence $
\det\tilde A_{n}=\det(\tilde A_{n-2}+\hat BC^{-1}B^T)\det C=a_{n-1\,n}^2\det(\tilde A_{n-2}+\hat BC^{-1}B^T)$.  

Direct computation shows that all the entries of $\hat BC^{-1}B^T$ vanish except the first row which equals $(a_{n-1\,n}^{-1}-1)(a_{1\,n-1},a_{2\,n-1},\ldots,a_{n-2\,n-1})$, therefore
\[\begin{aligned}&\det(\tilde A_{n-2}+\hat BC^{-1}B^T)\\&=\det\tilde A_{n-2}+(a_{n-1\,n}^{-1}-1)\left|\begin{array}{cccccc}
a_{1\,n-1}&a_{2\,n-1}&\cdots&\cdots&\cdots&a_{n-2\,n-1}\\
a_{12}&0&-a_{23}&\cdots&\cdots&-a_{2\,n-2}\\
a_{13}&a_{23}&0&-a_{34}&\cdots&-a_{3\,n-2}\\
\vdots&&&&&\vdots\\
a_{1\,n-2}&a_{2\,n-2}&\cdots&\cdots&a_{n-3\,n-2}&0
\end{array}\right|\\\\
&=\det\tilde A_{n-2} +(a_{n-1\,n}^{-1}-1)\left|\begin{array}{cccccc}
a_{12}&0&-a_{23}&\cdots&\cdots&-a_{2\,n-2}\\
a_{13}&a_{23}&0&-a_{34}&\cdots&-a_{3\,n-2}\\
\vdots&&&&&\vdots\\
a_{1\,n-2}&a_{2\,n-2}&\cdots&\cdots&a_{n-3\,n-2}&0\\
a_{1\,n-1}&a_{2\,n-1}&\cdots&\cdots&a_{n-3\,n-1}&a_{n-2\,n-1}
\end{array}\right|,  \end{aligned}\]
where we used that $n$ is odd.  
Finally, in view of equation \eqref{eq:MatrixCon}, we can replace the 
matrix in the second determinant by an upper triangular matrix by performing appropriate 
column additions, obtaining 
\[\begin{aligned}
\frac{\det \tilde A_n}{a_{n-1, n}^2}&=\det\tilde A_{n-2} +(a_{n-1\,n}^{-1}-1)\left|\begin{array}{cccc}
a_{12}&&\multicolumn{2}{c}{\raisebox{-2.5ex}[0pt]{\Huge *}}\\
&a_{23}\\
\multicolumn{2}{c}{\raisebox{-4.5ex}[0pt]{\Huge 0}}&\ddots\\
&&&a_{n-2\,n-1}
\end{array}\right| \\
&=\det\tilde A_{n-2} +(a_{n-1\,n}^{-1}-1)a_{12}a_{23}\cdots a_{n-2\,n-1}>0.\end{aligned}\]

\end{proof}

By using the lemmas above, we can obtain the property of $m_i(t)$ at the time of blow-up.

\begin{theorem}\label{thm:BlowupOrder}
If $m_i$ blows up at some $t_0$ then $m_i$ has a pole of order $1$ at $t_0$. 
\end{theorem}
\begin{proof}
Since $m_i$'s are meromorphic in $t$ we can assume that the leading term in the Laurent series of $m_i$ around $t_0$ is $\frac{C_i}{(t-t_0)^{\alpha_i}}, \, C_i\neq 0$.  If the conclusion does not hold then
\[\alpha=\max_i\{\alpha_i\}\geq 2.\] Set $S=\{i_j:\alpha_{i_j}=\alpha\}=\{i_1,\ldots,i_k\}$ where $i_1<i_2<\cdots<i_k$ and $k$ is at least $2$ by virtue of Lemma \ref{lem:constants} 
with $p=1$.  Comparing the leading term of both sides of \eqref{eq:b-ode-short-m} with $i_j\in S$, one can see the leading term on the left hand side is $\frac{-\alpha C_{i_j}}{(t-t_0)^{\alpha+1}}$ while the leading term on the right hand side is
\[\frac{2C_{i_j}}{(t-t_0)^{2\alpha}}\sum_{l=1}^k \mathrm{sgn}(i_j-i_l)e^{-|x_{i_j}-x_{i_l}|}C_{i_l}.\]  Since $2\alpha>\alpha+1$, the coefficient of $(t-t_0)^{-2\alpha}$ must be zero, which leads to a homogeneous linear equations $A_kC=0$, where $A_k=(\sgn(j-l) a_{jl})$ is a $k\times k$ skew-symmetric matrix with $a_{jl}=e^{x_{i_j}-x_{i_l}}(1<j<l)$ and $C=(C_{i_1},\ldots,C_{i_k})^T$. Additionally one can also find
\[C_{i_1}+C_{i_2}+\cdots+C_{i_k}=0\] by comparing the leading term in $M_1$. It is clear that $A_k$ satisfies (\ref{eq:MatrixCon}) and the condition in Lemma
\ref{lem:OddDet}. Hence $C$ must be zero according to Corollary \ref{lem:EvenDet} and Lemma \ref{lem:OddDet}, which leads to a contradiction. \end{proof}

In the proof above, we only use that the $m_i(t)s$ are meromorphic. However, we can get stronger conclusions if we also take into account that the $x_i(t)s$ are holomorphic. 
\begin{theorem}\label{thm:leadingcoef}
If $m_{j_1},\ldots,m_{j_k}\,(1\leq j_1<\cdots<j_k\leq n)$ blow up at $t_c$ and all other $m_i$ remain bounded, then the following conclusions hold.

(1) $k$ must be even. 

(2) $t_c$ must be a collision time. Moreover for all odd $l$ such that $1\leq l<k$, the peakon with label $j_l$ must collide with the peakon with label $j_{l+1}$. 

(3) The leading term of $m_{j_s}(t)$ in the Laurent series around $t_c$ must have the form $\frac{(-1)^s}{2(t-t_c)}$ for all $1\leq s \leq k$.  
\end{theorem}
\begin{proof}

Assume that $m_{j_1},m_{j_2},\ldots,m_{j_k}$ blow up at $t_c$.  Since $M_1=m_1+m_2+\cdots+m_n$ is conserved, $k$ is at least $2$.  Moreover, by 
Theorem \ref{thm:BlowupOrder} the leading term in each $m_{i_j}$'s Laurent series has the form $\frac{C_j}{t-t_c}$. Hence, by equations \eqref{eq:b-ode-short-m}, the coefficients $C=(C_1,\ldots,C_{k})^T$ satisfy the linear equations
\begin{equation}\label{eq:leadingterm-m}
A_k C=-\frac12(1,\ldots,1)^T,
\end{equation}
where the matrix $A_k =(A^{(k)}_{lm})_{1\leq l,m\leq k}=(\mathrm{sgn}(l-m)
e^{-|x_{j_m}-x_{j_l}|})_{1\leq l,m\leq k}$ satisfies (\ref{eq:MatrixForm}) and (\ref{eq:MatrixCon}).  Likewise, comparing the leading terms of both sides of \eqref{eq:b-ode-short-x} with the subscript $j_s\,(1\leq s\leq k)$, one finds
\begin{equation}\label{eq:leadingterm-x}
B^{(k)}C=0,
\end{equation} 
where the matrix $B^{(k)}=(B^{(k)}_{lm})_{1\leq l,m\leq k}=(e^{-|x_{j_m}-x_{j_l}}|)_{1\leq l,m\leq k}$.
 Now we prove that the theorem holds for $k=2$. In this case, \eqref{eq:leadingterm-m} and \eqref{eq:leadingterm-x} reduce to
\[\left[\begin{array}{cc} 0&-a_{12}\\a_{12}&0\end{array}\right] \left[\begin{array}{c} C_1\\C_2\end{array}\right]= -\frac12\left[\begin{array}{c}1\\1\end{array}\right],\;\left[\begin{array}{cc} 1&a_{12}\\a_{12}&1\end{array}\right] \left[\begin{array}{c} C_1\\C_2\end{array}\right]= 0,\] where $a_{12}=e^{x_{j_1}-x_{j_2}}$. 
 Direct computation shows that the solution of the equations exists iff $a_{12}=1$, and the solution is $C_2=-C_1=\frac12$. Since $a_{12}=1$ is equivalent to $x_{j_1}(t_c)=x_{j_2}(t_c)$, we conclude that the peakon with label $j_1$ collides at $t_c$ with the peakon with label $j_2$. 
 Suppose now the conclusions are valid for $k-2$.  
We will show that they hold for $k$ as well. Let us use the same block decomposition 
as in equation \eqref{eq:An}, obtaining 
\[\begin{aligned}A_k&=\left[\begin{array}{ccc}A_{k-2}&-\mathbf{a}&-a_{k-1\,k}\mathbf{a}\\ \mathbf{a}^T&0&-a_{k-1\,k}\\ a_{k-1\,k}\mathbf{a}^T &a_{k-1\,k}&0\end{array}\right],\\\\B^{(k)}&=\left[\begin{array}{ccc}B^{(k-2)}&\mathbf{a}&a_{k-1\,k}\mathbf{a}\\ \mathbf{a}^T&1&a_{k-1\,k}\\ a_{k-1\,k}\mathbf{a}^T &a_{k-1\,k}&1\end{array}\right]
,\end{aligned}\] where $\mathbf{a}=(a_{1\,k-1},a_{2\,k-1},\ldots,a_{k-2\,k-1})^T$. Let us now 
combine the last two rows of \eqref{eq:leadingterm-m} and \eqref{eq:leadingterm-x}, writing them collectively as
\[\left[\begin{array}{ccc}\mathbf{a} ^T&0&-a_{k-1\,k}\\ a_{k-1\,k}\mathbf{a}^T &a_{k-1\,k}&0\\\mathbf{a}^T&1&a_{k-1\,k}\\ a_{k-1\,k}\mathbf{a}^T&a_{k-1\,k}&1\end{array}\right]C= \left[\begin{array}{c}-\frac12\\ -\frac12\\ 0\\0 \end{array}\right].  \] The latter expression can subsequently be easily reduced to 
\[\left[\begin{array}{ccc}\mathbf{a}^T&0&-a_{k-1\,k}\\ 0&a_{k-1\,k}&a_{k-1\,k}^2\\0&1&2a_{k-1\,k}\\ 0&0&1\end{array}\right]
\left[\begin{array}{c}C_1\\ \vdots\\ C_{k-2}\\C_{k-1}\\ C_{k} \end{array}\right]= 
\frac12\left[\begin{array}{c}-1\\ -(1-a_{k-1\,k})\\ 1\\ 1 \end{array}\right],\] which implies the condition for the existence of the solution to be $a_{k-1\,k}^2=1$, hence $a_{k-1,\,k}=1$. The latter condition indicates the collision of $m_{j_k}$ with $m_{j_{k-1}}$. Futhermore, 
 the solution for the last two components of $C$ is then $C_{k-1}=-\frac12$ and $C_{k}=\frac12$, which proves the sign statement for the last two components.  Substituting $a_{k-1\,k}=1, C_{k-1}=-\frac12, C_{k}=\frac12$ into \eqref{eq:leadingterm-m} and \eqref{eq:leadingterm-x} and denoting the 
 first $k-2$ components of $C$ by $\mathbf{C}$ we obtain the following equations: 
 \[A_{k-2} \mathbf{C}=-\frac12(1,\ldots,1)^T,\; B^{(k-2)}\mathbf{C}=0, \; \mathbf{a}^T\mathbf{C}=0.\] The first two equations hold by the induction hypothesis.  To show that the third equation 
 holds automatically if the induction hypothesis is satisfied, we observe that as the 
 result of collisions ($j_1$th mass collides with $j_2$th mass  etc.) $\mathbf{a}^T=(a_{1\, k-1}, a_{1\, k-1}, 
 a_{3\, k-1}, a_{3\, k-1},\cdots, a_{k-3\, k-1}, a_{k-3\, k-1})$, hence, indeed, the last 
 equation follows from the induction hypothesis.  

\end{proof}

The following amplification of item 2 in the above theorem is automatic.  
\begin{corollary} Suppose $m_{j_1},\ldots,m_{j_k}\,(1\leq j_1<\cdots<j_k\leq n)$ blow up at $t_c$ and all other $m_i$ remain bounded.  Then for all odd $l$ such that $1\leq l<k$, the peakon with label $j_l$ must collide at $t_c$ with the peakon with label $j_{l}+1$ (its neighbour to the right).
\end{corollary} 

So far we have established that when the masses become unbounded the collisions 
must occur.  The converse turns out to be valid as well.  
\begin{theorem}\label{thm:notriple} \mbox{} For all initial conditions for which $M_n\neq 0$
the following properties are valid: 
\begin{enumerate}
\item[(1)] If the $i$-th peakon collides with another peakon/peakons at $t_c$, $m_i(t)$ must blow up at $t_c$.

\item[(2)] For all $i$, $m_i(t)$ cannot change its sign. In  particular, $m_i(t)\neq 0$ for $t<t_c$.

\item[(3)] There are no multiple collisions.  
\item[(4)] The distance between colliding peakons 
has a simple zero at $t_c$.  
\end{enumerate} 
\end{theorem}
\begin{proof}
Without loss of  generality we can suppose that the labelling is chosen so that $x_1(t_c)=\cdots=x_{k_1}(t_c)<x_{k_1+1}(t_c)=x_{k_1+2}(t_c)=\cdots =x_{k_2}(t_c)<\cdots <x_{k_l}(t_c)$ 
for all colliding peakons at distinct positions $x_{k_1}(t_c)<x_{k_2}(t_c)<\cdots <x_{k_l}(t_c)$.  
Let us now denote the set indexing all colliding peakons by $I$.  Since $M_n=m_1m_2(1-e^{x_1-x_2})^2\cdots m_{n-1}m_n (1-e^{x_{n-1}-x_n})^2$ is conserved and nonzero, it is 
clear that some of the masses must become unbounded.  Let us denote 
the set of labels of those masses which blow up at $t_c$ by $J$.  By Theorem \ref{thm:BlowupOrder} any such mass corresponds to a colliding peakon; thus $J\subset I$.  Moreover, any such a mass has a simple pole at $t_c$.  
On the other hand, for each colliding peakons with adjacent indices $j$ and $j+1$,  
$t_c$ is a zero of $(1-e^{x_j-x_{j+1}})^2$ of order bounded from below by $2$. Thus the order of the zero of all such exponential factors appearing in $M_n$ is bounded from below by $2([k_1-1]+
[k_2-1]+\cdots [k_l-1])=2(\abs{I}-l)$, where $\abs{I}$ denotes the cardinality of $I$.  
Hence, since all unbounded masses have poles of order $1$, $\abs{J}\geq 2(\abs{I}-l)$ to ensure that $M_n\neq 0$.  
The maximum of $l$ occurs when the masses collide in pairs, hence $l\leq \frac{\abs{I}}{2}$ and thus $\abs{J}\geq 2(\abs{I}-\frac{\abs{I}}{2})=\abs{I}$.  This proves that $J=I$ since $J\subset I$ and thus (1) is proven.  
To prove (3) we return to the inequality above which now reads $\abs{I}\geq2(\abs{I}-l)$, 
implying $l\geq \frac{\abs{I}}{2}$.  Since the right hand side is the maximum of $l$, 
$l=\frac{\abs{I}}{2}$ follows, which in turn implies that all collisions occur in pairs, hence absence of multiple collisions.  
To prove (4) we note that for $M_n$ to remain bounded the order of the zero of 
all exponential factors has to be exactly $\abs{I}=2\frac{\abs{I}}{2}$  hence 
each factor $(1-e^{x_j-x_{j+1}})^2$ has zero of order exactly equal $2$.

This concluded the proof of (1), (3) and (4).  
In order to prove (2) we suppose that for some $i$, $m_i(t)$ changes its sign, then there exists some $t_0$ for which $m_i(t_0)=0$ while all $m_j$s remain bounded since $t_0<t_c$.  
Hence
\[M_n=\lim_{t\to t_0} M_n=\lim_{t\to t_0}m_1m_2\cdots m_n(1-e^{x_1-x_2})^2\cdots(1-e^{x_{n-1}-x_n})^2=0.\] This contradicts $M_n\neq0$.
\end{proof}

\begin{remark}
It is now not difficult to verify that the constants of motion $M_1,\ldots,M_n$ 
can be extended up until the collision time $t_c$ by using Lemma \ref{lem:constants}, 
followed by Theorems \ref{thm:leadingcoef}, \ref{thm:notriple}.  
\end{remark}

\begin{corollary}\label{thm:SignbeforeCollision}
If $m_j$ and $m_{j+1}$ collide at $t_c>0$ then $m_j>0$ and $m_{j+1}<0$ before the collision.
\end{corollary}
\begin{proof}
Since collisions only occur in pairs, the leading terms in $m_j$ and $m_{j+1}$'s Laurent series must be $\mp\frac{1}{2(t-t_c)}$ respectively. This implies 
\[\lim_{t\to t_c^-}m_j=+\infty,\quad \lim_{t\to t_c^-}m_{j+1}=-\infty.\] The conclusion holds since  in view of Theorem \ref{thm:notriple} $m_j$ and $m_{j+1}$ cannot change their signs.  
\end{proof}

The following proposition shows that the {\sl simultaneous collisions} 
(several peakon-antipeakon pairs collide at distinct locations at the common time $t_c$) can happen.  We indicate below how certain symmetric initial conditions will lead to 
simultaneous collisions. To this end we consider equations \eqref{eq:b-ode-short} for $n=4$ and a special choice of initial conditions.  

\begin{lemma}\label{lem:SymCon}
If the initial conditions satisfy
\begin{subequations}\label{con:symmetry}
\begin{align}
m_1(0)&=-m_4(0)>0,\; m_2(0)=-m_3(0)<0,\;\\
x_1(0)&-x_2(0)=x_3(0)-x_4(0)<0, \, \, 
\end{align}
\end{subequations}
then $m_1(t)=-m_4(t),\, m_2(t)=-m_3(t), \, x_1(t)-x_2(t)=x_3(t)-x_4(t) $ will hold for all $0<t<t_c$.
\end{lemma}
\begin{proof}
Consider the following ODEs
\[\begin{aligned}
&\dot x_1=m_1+m_2e^{x_1-x_2}-m_2e^{x_1-x_3}-m_1e^{2x_1-x_2-x_3},\\
&\dot x_2=m_1e^{x_1-x_2}+m_2-m_2e^{x_2-x_3}-m_1e^{x_1-x_3},\\
&\dot x_3=m_1e^{x_1-x_3}+m_2e^{x_2-x_3}-m_2-m_1e^{x_1-x_2}, \\
&\dot m_1=-2m_1(-m_2e^{x_1-x_2}+m_2e^{x_1-x_3}+m_1e^{2x_1-x_2-x_3}), \\
&\dot m_2=-2m_2(m_1e^{x_1-x_2}+m_2e^{x_2-x_3}+m_1e^{x_1-x_3}), \\
\end{aligned}\]
then direct computation shows that $\{x_1,x_2,x_3,x_2+x_3-x_1,m_1,m_2,-m_2,-m_1\}$ satisfy the system of ODEs (\ref{eq:b-ode-short}) for $n=4$.
\end{proof}
The following is then immediate (see figure \ref{fig:1}).   
\begin{corollary}
If the initial conditions (\ref{con:symmetry}) hold and the peakon-antipeakon pair 
$(m_1,m_2)$ collides at $t_c$ then so does $(m_3,m_4)$ and vice-versa.
\end{corollary}
\begin{center} 
\begin{figure}[heretp]
\includegraphics[height=6cm, width=12cm]{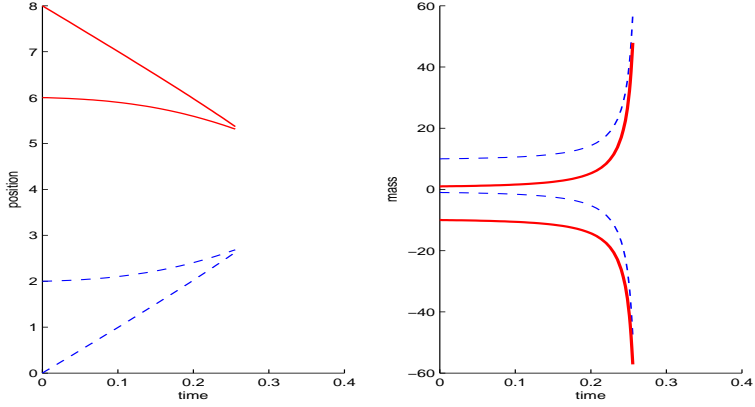}
\caption{Two symmetric peakon-antipeakon pairs with masses $m_1(0)=10=-m_4(0)$, $m_2(0)=-1=-m_3(0)$ undergo a simultaneous collision.  } \label{fig:1}
\end{figure}
\end{center}
\section{Collisions and shocks}\label{sec:shocks} 
In this section we investigate the behaviour of $m$ and $u$ at the time of collision(s).  
We start with $m$ and observe that since the collision of peakons occurs in pairs it is sufficient 
to study a fixed colliding pair $m_j, m_{j+1}$.    
\begin{theorem}\label{thm:shocks}
If $m_j$ collides with $m_{j+1}$ at time $t_c>0$ and the position $x_c$, then
\[\begin{aligned}&\lim_{t\to t_c^-}(m_j(t)\delta(x-x_j(t))+m_{j+1}(t)\delta(x-x_{j+1}(t))) \\
=&\left(\lim_{t\to t_c^-}(m_j+m_{j+1})\right)\delta(x-x_c)+\frac12\left(\dot x_j(t_c)-\dot x_{j+1}(t_c)\right)\delta'(x-x_c)\\
=&\left(\lim_{t\to t_c^-}(m_j+m_{j+1})\right)\delta(x-x_c)+\frac12\left(\lim_{t\to t_c^-}(u(x_j(t),t)-u(x_{j+1}(t),t)\right)\delta'(x-x_c)\end{aligned}\]
in $\mathscr{D}'(\mathbb{R})$.
\end{theorem}
\begin{proof}
For an arbitrary $\varphi(x)\in\mathscr{D}(\mathbb{R})$,
\[\langle m_j(t)\delta(x-x_j(t))+m_{j+1}(t)\delta(x-x_{j+1}(t)),\varphi(x)\rangle=m_j(t)\varphi(x_j(t))+m_{j+1}(t)\varphi(x_{j+1}(t)).\]
Using Corollary \ref{thm:SignbeforeCollision} we can write
\[m_j=-\frac1{2(t-t_c)}+C_0+O(t-t_c),\quad m_{j+1}=\frac1{2(t-t_c)}+\tilde C_0+O(t-t_c) \] around $t_c$.  Hence 
\[\begin{aligned}&\lim_{t\to t_c^-}\langle m_j(t)\delta(x-x_j(t))+m_{j+i}(t)\delta(x-x_{j+1}(t)),\varphi(x)\rangle \\
=&(C_0+\tilde C_0)\varphi(x_c)-\lim_{t\to t_c}\frac{\varphi(x_j(t))-\varphi(x_{j+1}(t))}{2(t-t_c)} \\
=&\left(\lim_{t\to t_c^-}(m_j+m_{j+1})\right)\varphi(x_c)-\frac12\left(\lim_{t\to t_c^-}(\dot{x}_j-\dot{x}_{j+1})\right)\varphi'(x_c) \\
=&\left(\lim_{t\to t_c^-}(m_j+m_{j+1})\right)\varphi(x_c)-\frac12\left(\lim_{t\to t_c^-}(u(x_j(t),t)-u(x_{j+1}(t),t))\right)\varphi'(x_c),
\end{aligned}\] where in the last step we used equation \eqref{eq:b-ode-short-x}.  
The conclusion now follows from the definition of $\delta$ and $\delta '$.  
\end{proof}
Since $m=u-u_{xx}$ we have the immediate corollary. 
\begin{corollary}\label{cor:shocks} 
Suppose $m(0)=2 \sum_{k=1}^n m_k(0)\delta(x-x_k(0))$ is a multipeakon at $t=0$ for which 
$M_n\neq 0$ and such that 
at $t_c$ one, or several of its peakon-antipeakon pairs collide.  For any colliding 
pair $k, k+1$ let us denote $\lim_{t\to t_c^-} u(x_k(t))=u_l(x_k(t_c)), \; \lim_{t\to t_c^-} u(x_{k+1}(t)=u_r(x_k(t_c))$ respectively.  
Then $$\lim_{t\to t_c^-} m(t)=2\sum_{k=1}^n \tilde m_k(t_c) \delta (x-x_k(t_c))+2\sum_{\substack{k: \text{pairs} \\x_k, x_{k+1}\text{collide}}}
s_k(t_c) \delta'(x-x_k(t_c))\; \text{ in }  \cal{D}'(\R). $$ The shock strengths are given by $$s_k(t_c)=\frac{u_l(x_k(t_c))-u_r(x_k(t_c))}{2}, $$ 
and they satisfy the (strict) entropy condition 
$s_k(t_c)>0$.  
\end{corollary} 
\begin{proof}
It suffices to prove the claim if there is only one colliding pair; the general case follows easily 
since masses collide pairwise.  For $t<t_c$ the measure evolves as $m(t)=2\sum_{k=1}^nm_k(t) \delta(x-x_k(t))$ where 
$x_k(t), m_k(t) $ satisfy equations \eqref{eq:b-ode-short-x}, \eqref{eq:b-ode-short-m} respectively.  Suppose now that the pair $j, j+1$ collides at the point $x_c$. Then 
by Theorem \ref{thm:shocks} $\lim_{t\to t_c} m(t)=2\sum_{k\ne j,j+1} m_k(t_c) \delta(x-x_k(t_c))+2\lim_{t\to t_c^-} (m_j+m_{j+1}(t) \delta(x-x_c) +2\frac12 (\dot x_j(t_c)-\dot x_{j+1}(t_c))\delta'(x-x_c). $ To prove that $s_j(t_c)\geq 0$ we write $s_j(t_c)=\frac12 \lim_{t\to t_c^-} (\dot x_j(t)-\dot x_{j+1}(t))=\frac12 \lim_{t\to t_c^-}(u(x_j(t),t)-u(x_{j+1}(t),t))$ and observe
\begin{equation*}
\dot x_j(t_c)-\dot x_{j+1}(t_c)=\lim _{t\to t_c} \frac{x_{j+1}(t)-x_j(t)}{t_c-t} 
\end{equation*}
which implies the entropy condition $s_j(t_c)\geq 0$ in view of the ordering assumption 
$x_j(t) <x_{j+1}(t)$.  The strict inequality follows from item (4) in Theorem \ref{thm:notriple}.  
\end{proof} 
The following amplification of the previous theorem brings the issues of the 
wave breakdown and a shock creation sharply into focus. To put our result into the proper perspective we first review the well-posedness result for $L^1(\mathbb{R} )\cap BV(\mathbb{R})$ proven by Coclite and Karlsen (\cite{coclite-karlsen-DPwellposedness}, Section 3).  
We present only the core result pertinent to our paper.

\begin{theorem} [Coclite-Karlsen]
Let $u_0\in L^1(\mathbb{R})\cap BV (\mathbb{R})$.  Then there exists a unique entropy weak solution to the Cauchy problem $u|_{t=0}=u_0$ for the DP equation \eqref{eq:DP}.  
\end{theorem}
It is then proven in \cite{lundmark-shockpeakons} that the shockpeakon ansatz 
\begin{equation*}
u(x,t)=\sum_{j=1}^n \{m_j(t)-s_j(t) \sgn (x-x_j(t))\}e^{-|x-x_j(t)|}, 
\end{equation*}
is an entropy weak solution provided the shock strengths $s_j\geq 0$.  This sets the stage for the next theorem.

\begin{theorem}Assume that a multipeakon solution $u(x,t)$ exists on $\mathbb{R}^n\times[0,t_c)$, then $u(\cdot,t)\in L^1(\mathbb{R})\cap BV(\mathbb{R})$ for all $0\leq t<t_c$ and $u(\cdot,t)$ converges in $L^1$ to the shockpeakon 
\[u(x,t_c)=\sum_{i=1}^n\tilde m_i(t_c)e^{-|x-x_i(t_c)|}+\sum_{i=1}^nC_i\dot x_i(t_c)\mathrm{sgn}(x-x_i(t_c))e^{-|x-x_i(t_c)|},  \quad \] $u(\cdot,t_c)\in L^1(\mathbb{R})\cap BV(\mathbb{R})$, where $m_i(t)$'s Laurent expansion around $t_c$ is written as 
\[m_i(t)=\frac{C_i}{t-t_c}+\sum_{l=0}^\infty a_l(t-t_c)^l\stackrel{\mathrm{def}}{=}\frac{C_i}{t-t_c}+\tilde m_i(t), \] with the proviso that $C_i=0$ if the $i$th mass is not involved in a collision and 
$C_i=-\frac12$ for a colliding peakon, $C_i=\frac12$ for a colliding antipeakon, respectively.   
\end{theorem}
\begin{proof}
We start with the case $n=2$. Then $u(x,t)=m_1(t)e^{-|x-x_1(t)|}+m_2(t)e^{-|x-x_2(t)|}$ and $x_1(t_c)=x_2(t_c)=x_c$. According to Theorem \ref{thm:leadingcoef}, we have
\[m_1(t)=-\frac1{2(t-t_c)}+\tilde m_1(t),\quad m_2(t)=\frac1{2(t-t_c)}+\tilde m_2(t),\] where $\tilde m_1(t),\tilde m_2(t)$ are analytic around $t_c$. It is clear that \[\tilde m_1(t)e^{-|x-x_1(t)|}+\tilde m_2(t)e^{-|x-x_2(t)|}\in L^1(\mathbb{R})\cap BV(\mathbb{R}) \quad\text{for all $0\leq t\leq t_c$}.\]
By the mean value theorem we find that\[\begin{aligned}&v(x,t)\stackrel{\mathrm{def}}{=}\frac1{2(t-t_c)}(e^{-|x-x_2(t)|}-e^{-|x-x_1(t)|})\\
=&\left\{\begin{aligned}
&\frac12\left.\left[\dot x_1(s)e^{x-x_1(s)}-\dot x_2(s)e^{x-x_2(s)}\right]\right|_{s=t+\theta_1(t_c-t)},&&x<x_1(t)<x_2(t),\\
&\frac12\left.\left[\dot x_2(s)e^{x_2(s)-x}-\dot x_1(s)e^{x_1(s)-x}\right]\right|_{s=t+\theta_2(t_c-t)},&&x_1(t)<x_2(t)<x,\\
&\frac{e^{x-x_c}-e^{x_c-x}}{2(t-t_c)}-\left.\frac12(\dot x_1(s)e^{x_1(s)-x}+\dot x_2(s)e^{x-x_2(s)})\right|_{s=t+\theta_3(t_c-t)},&&x_1(t)<x<x_2(t), 
\end{aligned}\right.
\end{aligned} \]
where $0<\theta_j<1, \; j=1,2,3$. Hence we have the pointwise limit
\[\lim_{t\to t_c^-}v(x,t)=\begin{cases} \mathrm{sgn}(x-x_c)(-\frac12\dot x_1(t_c)+\frac12\dot x_2(t_c))e^{-|x-x_c|}, \;  \text{ for } x\neq x_c, \\
-\frac{\dot x_1(t_c)+\dot x_2(t_c)}{2}, \;  \text{ for } x=x_c.  \end{cases} 
\]
Let us define
\begin{equation*}
v(x,t_c)=\begin{cases} \mathrm{sgn}(x-x_c)(-\frac12\dot x_1(t_c)+\frac12\dot x_2(t_c))e^{-|x-x_c|}, \;  \text{ for } x\neq x_c, \\
0, \;  \text{ for } x=x_c.  \end{cases} 
\end{equation*}
 and consider the integral 
\[\int\limits_{-\infty}^{+\infty}|v(x,t)-v(x,t_c)|\mathrm{d}x= \left(\int\limits_{-\infty}^{x_1(t)}+\int\limits_{x_1(t)}^{x_2(t)}+\int\limits_{x_2(t)}^{+\infty}\right)|v(x,t)-v(x,t_c)|\mathrm{d}x. \] Then the first and the last term of the right hand side converge to zero as $t\to t_c^-$ due to Lebesgue's dominated convergence theorem. 

Observe that the second term satisfies \[\begin{aligned}
&\int\limits_{x_1(t)}^{x_2(t)}|v(x,t)-v(x,t_c)|\mathrm{d}x\leq\int\limits_{x_1(t)}^{x_2(t)}|v(x,t)|\mathrm{d}x+\int\limits_{x_1(t)}^{x_2(t)}|v(x,t_c)|\mathrm{d}x\\
\leq &\int\limits_{x_1(t)}^{x_2(t)}|v(x,t)|\mathrm{d}x+\int\limits_{x_1(t)}^{x_2(t)}\frac{|e^{x-x_2(t)}-e^{x_c-x}|}{2(t_c-t)}\mathrm{d}x +\frac12\int\limits_{x_1(t)}^{x_2(t)}|\dot x_1(s)|e^{x_1(s)-x}\mathrm{d}x+\frac12\int\limits_{x_1(t)}^{x_2(t)}|\dot x_2(s)|e^{x-x_2(s)}\mathrm{d}x\\
= &\int\limits_{x_1(t)}^{x_2(t)}|v(x,t)|\mathrm{d}x+\frac{x_2(t)-x_1(t)}{2(t_c-t)}|e^{y-x_c}-e^{x_c-y}|+\frac12\int\limits_{x_1(t)}^{x_2(t)}
|\dot x_1(s)|e^{x_1(s)-x}\mathrm{d}x+\frac12\int\limits_{x_1(t)}^{x_2(t)}|\dot x_2(s)|e^{x-x_2(s)}\mathrm{d}x\\
\end{aligned}\] where $s\in (t,t_c)$ and $y\in (x_1(t),x_2(t))$. Since $|\dot x_1(s)|e^{x_1(s)-x}$ and $|\dot x_2(s)|e^{x-x_2(s)}$ are bounded, and 
$$x_2(t)-x_1(t)\to 0, \quad 
\frac{x_2(t)-x_1(t)}{2(t_c-t)}\to \dot x_1(t_c)-\dot x_2(t_c),\quad |e^{y-x_c}-e^{x_c-y}|\to 0$$
as $t\to t_c^-$, we have that $v(x,t)$ converges to $v(x,t_c)$ in the sense of $L^1$, which shows that the conclusion holds for $n=2$.

In general, since collisions can only occur in pairs, we can assume that $m_{j_1}(t),m_{j_1+1}(t),m_{j_2}(t),m_{j_2+1}(t),\ldots,m_{j_k}(t),m_{j_k+1}(t)$ blow up at $t_c$ and all the other $m_i$'s remain bounded. It is clear that
$m_i(t)e^{-|x-x_i(t)|}$ lies in $L^1(\mathbb{R})\cap BV(\mathbb{R})$ and converges to $m_i(t_c)e^{-|x-x_i(t_c)|}$ in $L^1$ if $m_i(t)$ remains bounded at $t_c$. Meanwhile, according the proof above, we can easily see that 
\[m_{j_s}(t)e^{-|x-x_{j_s}(t)|}+m_{j_s+1}(t)e^{-|x-x_{j_s+1}(t)|}\in L^1(\mathbb{R}) \cap BV(\mathbb{R}) \quad\text{for all $1\leq s\leq k$},\] whose limit is 
\[\begin{aligned}&\tilde m_{j_s}(t_c)e^{-|x-x_{j_s}(t_c)|}+\tilde m_{j_s+1}(t_c)e^{-|x-x_{j_s+1}(t_c)|}\\
&-\frac12\sgn(x-x_{j_s}(t_c))(\dot x_{j_s}(t_c)-\dot x_{j_s+1}(t_c))e^{-|x-x_j(t_c)|}\end{aligned}\] as $t\to t_c^-$, which leads to the conclusion.
\end{proof}
  
\begin{corollary}
The limit of a multipeakon $u(\cdot,t) $ for 
$t \rightarrow t_c^-$ has a unique entropy weak extension which is a shockpeakon in the sense of 
H. Lundmark.
\end{corollary}

\section{Acknowledgments}

We thank Hans Lundmark for numerous perceptive comments.  
J. S.  would like to thank the Centro Internacional de Ciencias (CIC)  in Cuernavaca (Mexico) for hospitality and F. Calogero for making the stay so enjoyable and productive.  
This work was supported by National Natural Science
Funds of China \newline
 [NSFC10971155 to L.Z.]; and Natural Sciences and Engineering Research Council of Canada
 [NSERC163953 to J.S.].  Both authors would like to thank the Department of Mathematics and Statistics of the University of Saskatchewan for making the collaboration possible.

\end{document}